\documentclass[sigconf]{acmart}
\usepackage{paralist}
\usepackage{url}
\usepackage{hyperref}

\setcopyright{acmlicensed}
\copyrightyear{2017}
\acmYear{2017}
\acmConference{ISSAC '17}{July 25-28, 2017}{Kaiserslautern, Germany}
\acmPrice{15.00}
\acmDOI{10.1145/3087604.3087611}
\acmISBN{978-1-4503-5064-8/17/07}

\begin{document}



\providecommand{\QQ}{\mathbb{Q}}
\providecommand{\ZZ}{\mathbb{Z}}
\providecommand{\NN}{\mathbb{N}}
\providecommand{\RR}{\mathbb{R}}
\providecommand{\CC}{\mathbb{C}}

\makeatletter
\def\Ddots{\mathinner{\mkern1mu\raise\p@
\vbox{\kern7\p@\hbox{.}}\mkern2mu
\raise4\p@\hbox{.}\mkern2mu\raise7\p@\hbox{.}\mkern1mu}}
\makeatother

\title[Nemo/Hecke: Computer Algebra and Number Theory Packages for Julia]{Nemo/Hecke: Computer Algebra and Number Theory Packages for the Julia Programming Language}

\author{Claus Fieker}
\affiliation{%
  \institution{TU Kaiserslautern}
  \streetaddress{Fachbereich Mathematik, Postfach 3049}
  \postcode{67653}
  \city{Kaiserslautern} 
  \country{Germany} 
}
\email{fieker@mathematik.uni-kl.de}

\author{William Hart}
\affiliation{%
  \institution{TU Kaiserslautern}
  \streetaddress{Fachbereich Mathematik, Postfach 3049}
  \postcode{67653}
  \city{Kaiserslautern} 
  \country{Germany} 
}
\email{goodwillhart@gmail.com}

\author{Tommy Hofmann}
\affiliation{%
  \institution{TU Kaiserslautern}
  \streetaddress{Fachbereich Mathematik, Postfach 3049}
  \postcode{67653}
  \city{Kaiserslautern} 
  \country{Germany} 
}
\email{thofmann@mathematik.uni-kl.de}

\author{Fredrik Johansson}
\affiliation{%
  \institution{Inria Bordeaux \& Institut de Math\'{e}matiques de Bordeaux}
  \streetaddress{200 Avenue de la Vieille Tour}
  \postcode{33400}
  \city{Talence} 
  \country{France} 
}
\email{fredrik.johansson@gmail.com}


\begin{abstract}
We introduce two new packages, Nemo and Hecke, written in the Julia programming language
for computer algebra and number theory.
We demonstrate that high performance generic
algorithms can be implemented in Julia, without the need to resort to a low-level C
implementation.
For specialised algorithms, we use Julia's efficient
native C interface to wrap
existing C/C++ libraries such
as Flint, Arb, Antic and Singular.
We give examples of how to use Hecke and Nemo and discuss
some algorithms that we have implemented to provide high performance basic
arithmetic.
\end{abstract}

\maketitle

\section{Introduction}

Nemo\footnote{\url{http://nemocas.org}. Nemo and Hecke are BSD licensed.} is a  computer algebra package for the Julia programming language.
The eventual aim is
to provide highly performant commutative algebra, number theory, group theory and discrete
geometry routines.
Hecke is a Julia package that builds on Nemo to cover algebraic number theory.

Nemo consists of two parts: wrappers of specialised C/C++
libraries (Flint \cite{flint}, Arb \cite{arb}, Antic \cite{antic} and Singular
\cite{singular}), and
implementations of generic algorithms and mathematical data
structures in the Julia language. So far the fully recursive, generic constructions include
univariate and multivariate
polynomial rings, power series rings, residue rings (modulo principal ideals),
fraction fields, and matrices.

We demonstrate
that Julia is effective for implementing high performance computer
algebra algorithms. Our implementations also
include a number of
improvements over the state of the art
in existing computer algebra systems.


\section{Computer algebra in Julia}
\label{sect:domains}

Julia \cite{julia} is a modern programming language designed
to be both performant and flexible.
Notable features include an innovative type system,
multiple dispatch,
just-in-time (JIT) compilation
supported by dynamic type inference,
automatic memory management (garbage collection),
metaprogramming capabilities,
high performance builtin collection types (dictionaries, arrays, etc.),
a powerful standard library,
and a familiar imperative syntax (like Python).
Julia also has an efficient native C interface, and more recently a C++ interface.
In addition, Julia provides an interactive console and can be used with Jupyter notebooks.

Julia was designed with high performance numerical algorithms in mind, and provides near
C or Fortran performance for low-level arithmetic. This allows low-level algorithms to be
implemented in Julia itself. However, for us, the main advantage of Julia is its ability
to JIT compile generic algorithms for specific types, in cases where a specific
implementation does not exist in C.
 
One of the most important features from a mathematical point of view is Julia's parametric type system.
For example, in Julia, a 1-dimensional array of integers has type \texttt{Array\{Int, 1\}}.
We make heavy use of the parametric type system in Nemo and Hecke. Generic polynomials
over a ring $R$ have type \texttt{GenPoly\{T\}} where $T$ is the type of elements of the ring $R$.
Parametric types bear some resemblance to C++ template classes, except that in Julia the type $T$ can
be constrained to belong to a specified class of types.

\subsection{Modelling domains in Julia}

Julia provides two levels of types: abstract and concrete types.
Concrete types are like types in Java, C or C++,
etc. Abstract types can be thought of as collections of types. They are used when writing generic
functions that should work for any type in a collection.

To write such a generic function, we first create an abstract type, then we create the individual
concrete types that belong to that abstract type. The generic function is then specified with a type
parameter, \emph{T} say, belonging to the abstract type.

In Julia, the symbol <: is used to specify that a given type belongs to a given abstract type.
For example the built-in Julia type Int64 for 64-bit machine integers belongs to the Julia
abstract type Integer.

Abstract types in Julia can form a hierarchy. For example, the Nemo.Field abstract type belongs
to the Nemo.Ring abstract type. An object representing a field in Nemo has type belonging to
Nemo.Field. But because we define the inclusion Nemo.Field <: Nemo.Ring, the type of such an
object also belongs to Nemo.Ring. This means that any generic function in Nemo which is designed
to work with ring objects will also work with field objects.

Julia always picks the most specific function that applies to a given type. This allows one to
implement a function at the most general level of a type hierarchy at which it applies. One can
also write a version of a given function for specific concrete types. For example, one may wish
to call a specific C implementation for a multiplication algorithm, say, when arguments with a
certain very specific given type are passed.

Another way that we make use of Julia's abstract type system in Nemo/Hecke is to distinguish
between the type of elements of fields, and the fields themselves, and similarly for all other
kinds of domains in Nemo. 

Figure~\ref{fig:types} shows the abstract type hierarchy in Nemo.

\begin{figure}[h]
\centering
\includegraphics[scale=0.19]{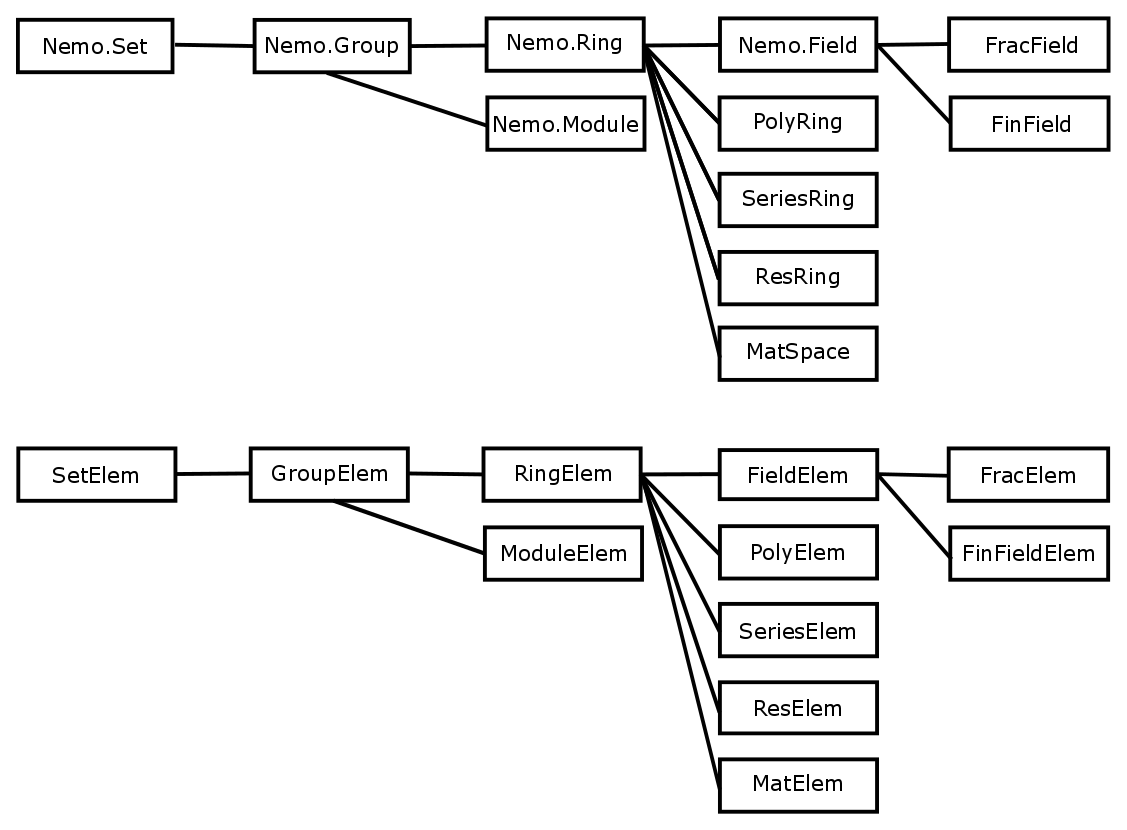}
\caption{The Nemo abstract type hierarchy}
\label{fig:types}
\end{figure}

Naively, one may expect that specific mathematical domains in Nemo/Hecke can be modeled as types
and their elements as objects of the given type. But there are various reasons why this is not a
good model.

As an example, consider the ring $R = \mathbb{Z}/n\mathbb{Z}$. If we were to model the ring $R$
as a type, then the type would need to contain information about the modulus $n$. This is not
possible in Julia if $n$ is an object, e.g.\ a multiprecision integer.
Further, Julia dispatches on type, and each time
we call a generic function with values of a new type, the function is recompiled by the
JIT compiler for that new type. This would result in very poor performance if we were writing a
multimodular algorithm, say, as recompilation would be triggered for each distinct modulus $n$.
For this reason, the modulus $n$ needs to be attached to the elements of the ring, not to the type
associated with those elements.

The way we deal with this is to have special (singleton) objects, known as parent objects, that act
like types, but are in fact ordinary Julia objects. As ordinary objects, parents can contain
arbitrary information, such as the modulus $n$ in $\mathbb{Z}/n\mathbb{Z}$. Each object representing
an element of the ring then contains a pointer to this parent object.

This model of mathematical parent objects is taken from SageMath \cite{sage} which in turn followed
Magma \cite{magma}.

Julia allows ordinary objects to be made callable. We make use of the facility to write coercions
and constructors for elements of mathematical domains in Nemo/Hecke. For example, the following
code constructs $a = 3 \pmod{7}$.

\begin{small}
\begin{verbatim}
R = ResidueRing(ZZ, 7)
a = R(3)
\end{verbatim}
\end{small}

We also make use of the parent object system to encode information such as context objects needed
by C libraries. As Julia objects can have precisely the same bit representation as native C objects,
parent objects can be passed directly to C functions. Additional fields in these objects can be safely
appended if it is desirable to retain more information at the Julia level than the C/C++ level.

Nemo wraps C types provided by libraries such as Flint
for ground domains.
For example, Nemo's \texttt{ZZ} wraps Flint's \texttt{fmpz}
integer type.
Nemo also uses specialised C implementations for nested structures
where available.
For instance, \texttt{PolynomialRing(R, "x")} which constructs
the univariate polynomial ring $R[x]$ automatically switches to
a wrapper of Flint's \texttt{fmpz\_poly} when $R = \ZZ$
instead of using the Julia implementation designed for generic $R$.

\subsection{In-place operations}

C libraries such as Flint allow mutating
objects in-place, which improves performance
especially when making incremental updates
to large objects.
Nemo objects are ostensibly immutable,
but special mutation methods
are provided for use in critical spots.
Instead of writing
\begin{small}
\begin{verbatim}
s += a * b
\end{verbatim}
\end{small}
in Nemo, which creates an object for \texttt{a * b}
and then replaces \texttt{s} with yet another new object, we
create a reusable scratch variable \texttt{t} with the same type as \texttt{a}, \texttt{b} and \texttt{s} and use
\begin{small}
\begin{verbatim}
mul!(t, a, b)
addeq!(s, t)
\end{verbatim}
\end{small}
which creates no new objects and avoids making a copy of the data in \texttt{s}
that is not being modified.

In Julia, in-place operators save not only memory allocation and copying
overheads, but also
garbage collection costs, which may be substantial in the worst case.

\section{Nemo examples}
\label{sect:nemo}

We present a few Nemo code examples which double as
benchmark problems.

\subsection{Multivariate polynomials}

The Fateman benchmark
tests arithmetic in $\mathbb{Z}[x_1,\ldots,x_k]$
by computing
$f \cdot (f+1)$ where $f = (1 + x + y + z + t)^n$. In Nemo,
this is expressed by the following Julia code:

\begin{small}
\begin{verbatim}
R,x,y,z,t = PolynomialRing(ZZ, ["x","y","z","t"])
f = (1+x+y+z+t)^30
f*(f+1)
\end{verbatim}
\end{small}

Table~\ref{tab:fateman} shows timing results compared to
SageMath~7.4, Magma V2.22-5 and Giac-1.2.2
on a 2.2~GHz AMD Opteron 6174.
We used sparse multivariate polynomials in all
cases (this benchmark could also be performed
using nested dense univariate arithmetic, i.e.\ in $\mathbb{Z}[x][y][z][t]$).

Nemo (generic) is our Julia implementation of Johnson's sparse
multiplication algorithm which handles $R[x_1,\ldots,x_k]$
for generic coefficient rings $R$.
Flint (no asm) is a reimplementation of the same
algorithm in C for the special case $R = \mathbb{Z}$,
and Flint (asm) is an assembly optimised version of this code.
Finally, Flint (array) and Giac implement an algorithm designed
for dense polynomials, in which terms are accumulated
into a big array covering all possible exponents.

Table~\ref{tab:pearce} shows timing results on the Pearce
benchmark, which consists of computing $f \cdot g$ where
$f = (1 + x + y + 2z^2 + 3t^3 + 5u^5)^n$,
$g = (1 + u + t + 2z^2 + 3y^3 + 5x^5)^n$.
This problem is too sparse to use the big array method,
so only the Johnson algorithm is used.

\begin{table}
\center
\caption{Time (s) on the Fateman benchmark.}
\begin{small}
\setlength{\tabcolsep}{2pt}
\renewcommand{\arraystretch}{0.95}
\begin{tabular}{c c c c c c c c} \hline
$n$ & Sage- & Magma & Nemo & Flint & Flint & Flint & Giac \\ 
    &   Math    &       & (generic) & (no asm) & (asm) & (array) & \\ \hline
     5  &  0.008 & $\sim$0.01   & 0.004      &        0.002    &        0.002       &      0.0003    &    0.0002 \\
    10  &  0.53  &   0.12  &   0.11     &         0.04    &         0.02       &         0.005   &      0.006 \\
    15  &   10   &    1.9  &      1.6   &         0.53    &          0.30      &           0.08   &         0.11 \\
    20  &   76   &    16   &    14.3    &          6.3    &            2.8     &            0.53   &         0.62 \\
    25  &  426   &   98    &      82    &          39     &         17.4        &           2.5     &         2.8 \\
    30  & 1814   & 439     &  362       &         168     &            82        &            11     &          14 \\
\end{tabular}
\label{tab:fateman}
\end{small}
\end{table}

By default, Nemo uses Flint for multivariate polynomials
over $R = \ZZ$ instead of the generic
Julia code~\footnote{The new Flint type \texttt{fmpz\_mpoly} is presently available in the git version of Flint.
The \texttt{mul\_johnson} (with and without asm enabled) and \texttt{mul\_array} methods were timed via Nemo.}.
We have timed both versions here
to provide a comparison between Julia and C.
It is somewhat remarkable that the generic Julia code comes
within a factor two of our C version in Flint (without asm),
and even runs slightly faster than the C code in Magma.

\begin{table}
\center
\caption{Time (s) on the Pearce benchmark.}
\begin{small}
\setlength{\tabcolsep}{2pt}
\renewcommand{\arraystretch}{0.95}
\begin{tabular}{c c c c c c c} \hline
$n$ & SageMath & Magma & Nemo & Flint & Flint & Giac \\ 
    &          &       & (generic) & (no asm) & (asm) & \\ \hline
     4  &   0.01 & $\sim$0.01  &   0.008          & 0.003         &   0.002             &                     0.004 \\
     6  &   0.20  &   0.08   &     0.07        &     0.03       &       0.03         &                           0.03 \\
     8  &   2.0  &     0.68   &     0.56       &      0.16      &        0.15        &                            0.28 \\
    10  &   11   &     3.0    &       2.2      &       0.77     &         0.71       &                             1.45 \\
    12  &   57   &    11.3    &      8.8       &        3.7     &           3.2      &                                4.8 \\
    14  &  214   &   36.8     &   32.3         &       16       &         12         &                               14 \\
    16  &  785   &   94.0     &   85.5         &       44       &         32         &                               39 \\
\end{tabular}
\label{tab:pearce}
\end{small}
\end{table}

\subsection{Generic resultant}

The following Nemo example code computes a resultant
in the ring $((\operatorname{GF}(17^{11})[y])/(y^3 + 3xy + 1))[z]$,
demonstrating generic recursive rings and polynomial
arithmetic. Flint is used for univariate polynomials over a finite field.

\begin{small}
\begin{verbatim}
R, x = FiniteField(17, 11, "x")
S, y = PolynomialRing(R, "y")
T = ResidueRing(S, y^3 + 3x*y + 1)
U, z = PolynomialRing(T, "z")

f = (3y^2+y+x)*z^2 + ((x+2)*y^2+x+1)*z + 4x*y + 3
g = (7y^2-y+2x+7)*z^2 + (3y^2+4x+1)*z + (2x+1)*y + 1
s = f^12
t = (s + g)^12
resultant(s, t)
\end{verbatim}
\end{small}

This example takes 179907~s in SageMath~6.8, 82 s in Magma V2.21-4
and 0.2 s in Nemo~0.4.

\subsection{Generic linear algebra}

We compute the determinant of a random $80\times80$ matrix
with entries in a cubic number field.
This benchmark tests generic linear algebra over a field
with coefficient blowup. The number field arithmetic is
provided by Antic.

\begin{small}
\begin{verbatim}
QQx, x = PolynomialRing(QQ, "x")
K, a = NumberField(x^3 + 3*x + 1, "a")
M = MatrixSpace(K, 80, 80)()

for i in 1:80
 for j in 1:80
   for k in 0:2
     M[i, j] = M[i, j] + a^k * (rand(-100:100))
   end
 end
end

det(M)
\end{verbatim}
\end{small}

This takes 5893~s in SageMath~6.8, 21.9~s in Pari/GP~2.7.4 \cite{pari}, 5.3~s in Magma V2.21-4,
and 2.4~s in Nemo~0.4.

\section{Generic algorithms in Nemo}
\label{sect:generic}

Many high level algorithms in number theory and computer algebra rely on
the efficient implementation of fundamental algorithms.
We next discuss some of the algorithms
used in Nemo for generic polynomials and matrices.

\subsection{Polynomial algorithms}

For generic coefficient rings, Nemo implements dense univariate polynomials,
truncated power series (supporting both relative and absolute precision models),
and multivariate polynomials in sparse distributed format.

The generic polynomial resultant code in Nemo uses subresultant
pseudoremainder sequences. This code can even be called for polynomials over a
ring with zero divisors, so long as impossible inverses are caught. The exception
can be caught using a try/catch block in Julia and a fallback method called, e.g.\
the resultant can be computed using Sylvester matrices. This allows an algorithm
with quadratic complexity to be called generically, with an algorithm of cubic
complexity only called as backup.

We make use of variants of the sparse algorithms described by Monagan and
Pearce for heap-based multiplication \cite{heapmul}, exact division and division with
remainder \cite{heapdiv} and powering~\cite{heappow}. For powering of polynomials with
few terms, we make use of the multinomial formula.
For multivariate polynomials over $\mathbb{Z}$, we used the generic Julia
code as a template to implement a more optimised C version in Flint.

In the case of pseudodivision for sparse, multivariate polynomials, we extract one
of the variables and then use a heap-based, univariate pseudoremainder
algorithm inspired by an unpublished manuscript of Monagan and Pearce.

We make use of this pseudoremainder implementation to implement the subresultant
GCD algorithm. To speed the latter up, we make use of a number of tricks employed by the
Giac/XCAS system \cite{giac}, as taught to us by Bernard Parisse. The most important
of these is cheap removal of the content that accumulates during the subresultant
algorithm by analysing the input polynomials to see what form the leading coefficient
of the GCD can take.

We also use heuristics to determine which permutation of the variables will lead to
the GCD being computed in the fastest time. This heuristic favours the variable with
the lowest degree as the main variable for the computation, with the other variables 
following in increasing order of degree thereafter. But our heuristic heavily favours
variables in which the polynomial is monic.

\subsection{Matrix algorithms}

For computing the determinant over a generic commutative ring, we implemented a generic
algorithm making use of Clow sequences \cite{clow} which uses only $O(n^4)$ ring operations
in the dimension $n$ of the matrix. However, two other approaches seem to always outperform
it. The first approach makes use of a generic fraction free LU decomposition. This algorithm
may fail if an impossible inverse is encountered. However, as a fallback, we make use of a
division free determinant algorithm. This computes the characteristic polynomial and then
extracts the determinant from that.

For determinants of matrices over polynomial rings, we use an interpolation approach.

We implemented an algorithm for computing the characteristic polynomial due to
Danilevsky (see below) and an algorithm that is based on computing the Hessenberg
form. We also make use of a generic implementation of the algorithm of Berkowitz which
is division free. 

As is well known, fraction free algorithms often improve the performance of matrix
algorithms over an integral domain. For example, fraction free Gaussian elimination
for LU decomposition, inverse, determinant and reduced row echelon form
are well-known.

We have also been able to use this strategy in the computation
of the characteristic polynomial. We have adapted the well-known 1937 method
of Danilevsky \cite{danilevsky} for computing the characteristic polynomial, into
a fraction-free version.

Danilevsky's method works by applying similarity transforms to reduce the
matrix to Frobenius form. Normally such computations are done over a
field, however each of the outer iterations in the algorithm introduce only
a single denominator. Scaling by this denominator allows us to avoid
fractions.
The entries become larger
as the algorithm proceeds because of the scaling, but conveniently
it is possible to prove that the introduced scaling factor can be removed 
one step later in the algorithm. This is an exact division and does not
lead to denominators.

Removing such a factor has to be done in a way that respects the
similarity transforms. We achieve this by making two passes over the matrix
to remove the common factor.

\subsubsection{Minimal polynomial}

Next we describe an algorithm we have implemented for efficient
computation of the minimal polynomial of an $n\times n$ integer matrix $M$.
A standard approach to this problem is known as ``spinning'' \cite{steel}.
We provide a brief summary of the main ideas, to fix notation and then
describe our contribution, which is a variant making use of Chinese
remaindering in a provably correct way.

We first summarise the theory for matrices over a field~$K$. The theory
relies on the following result, e.g.\ \cite{vinberg}.

\begin{theorem}
Suppose $M$ is a linear operator on a $K$-vector space $V$, and that
$V = W_1 + W_2 + \cdots + W_n$ for invariant subspaces $W_i$. Then the
minimal polynomial of $M$ is LCM$(m_1, m_2, \ldots, m_n)$, where $m_i$ is
the minimal polynomial of $M$ restricted to $W_i$.
\end{theorem}

The subspaces $W_i$ we have in mind are the following.

\begin{definition}
Given a vector $v$ in a vector space $V$ the \emph{Krylov subspace} $K(V, v)$
associated to $v$ is the linear subspace spanned by $\{v, Mv, M^2v, \ldots\}$.
\end{definition}

Krylov subspaces are invariant subspaces and so these results lead to an
algorithm for computing the minimal polynomial as follows.

Start with $v_1 = e_1 = (1, 0, \ldots, 0)$ and let $W_1 = K(V, v_1)$.
For each $i > 1$ let $v_i$ be the first standard basis vector $e_j$ that is
linearly independent of $W_1 + W_2 + \cdots + W_{i-1}$. Set
$W_i = K(V, v_i)$. 
By construction, $V = W_1 + W_2 + \cdots + W_n$, and the minimal polynomial
of $M$ is the least common multiple of the minimal polynomials $m_i$ of
$M$ restricted to the $W_i$.

The minimal polynomials $m_i$ are easy to compute.
For, if $v$, $Mv$, $M^2v, \ldots, M^{d-1}v$ is a basis for $W_i$,
there is a linear relation 
$$M^dv + c_{d-1}M^{d-1}v + \cdots + c_1Mv + c_0v = 0,$$
for some $c_i \in K$, and no smaller such relation. Letting
$$m_i(T) = c_0 + c_1T + \cdots + T^d,$$ we have $m_i(M)(v) = 0$.
One sees that $m_i(M)(f(M)v) = 0$ for all polynomials $f(T) \in K[T]$
and so $m_i(M)$ annihilates all of $W_i$. Thus $m_i$ is the required
minimal polynomial. 
 
To efficiently implement this algorithm, we keep track of and
(incrementally) reduce a matrix $B$ whose rows consist of all the
linearly independent vectors from the Krylov sequences computed so far.
Any column without a pivot in this reduced matrix $B$ corresponds to a
standard basis vector independent of the Krylov subspaces computed
so far. As each new Krylov subspace is computed, we append the corresponding
vectors to the matrix $B$, and reduce them.

It is also possible to define the minimal polynomial of a matrix $M$ over
$\mathbb{Z}$, since the null ideal $N_D(M) = \{f(T) \in D[T] \;|\; f(M) = 0\}$ of a
matrix $M$ over an integrally closed domain $D$ is principal \cite{brown}.

In the case $D = \mathbb{Z}$, we can define the minimal polynomial $m(T)$ of $M$
to be a (primitive) generator of this ideal. We have that $m(T)$ is monic since
the characteristic polynomial of $M$ is monic and an element of the null ideal.
By Gauss' Lemma for polynomials, this argument also shows that the minimal
polynomial of $M$ is the same as that of $M$ considered as a matrix over
$\mathbb{Q}$.

We have been informed by personal communication that the algorithm used by
Magma \cite{magma} for minimal polynomial computation over $\mathbb{Z}$ uses
spinning, however it uses a $p$-adic approach. Unfortunately
we have been unable to find a publication outlining their method.

We describe a multimodular variant of the spinning approach. The idea is to
reduce the matrix $M$ modulo many small primes $p$ and apply the method
described above over the field $\mathbb{Z}/p\mathbb{Z}$, for each prime $p$.
We then combine the minimal polynomials modulo the various primes $p$ using
Chinese remaindering.

The minimal polynomial of the reduction $M_{(p)}$ of $M$ modulo $p$ is the
reduction modulo $p$ of the minimal polynomial $m(T)$ of $M$ for all but
finitely many ``bad'' primes (see \cite{giesbrecht} Lemma 2.3). Bad primes
are detected if the degrees of the minimal polynomials modulo $p$ change at
any point during the algorithm.

Whilst bounds on the number of primes required in the Chinese remaindering
exist, e.g.\ based on Ovals of Cassini, these bounds are typically extremely
pessimistic. It is also unfortunately too expensive to simply evaluate the
minimal polynomial $m(T)$ at the matrix $M$ and compare with zero.

We obtain a better termination criterion by allowing a small amount of
information to 'leak' from the modulo~$p$ minimal polynomial computations.
Namely, for one of the (good) primes $p$, we record which standard
basis vectors $v_i$ were used to generate the Krylov subspaces $W_i$ when
computing the minimal polynomial of $M_{(p)}$. Recall that
$V = W_1 + W_2 + \cdots + W_n$, where $M_{(p)}$ is thought of as a linear
operator on $V$.

Thinking of $M$ as a linear operator on a $\mathbb{Q}$-vector space $V'$,
and letting $W'_i = K(V', v'_i)$, where $v'_i$ is the lift of $v_i$ to
$\mathbb{Q}$, it is clear that $V' = W'_1 + W'_2 + \cdots + W'_n$.

Thus, if $m(T)$ is believed to be the minimal polynomial of $M$,
e.g.\ because the Chinese remaindering has stabilised, then if $m(M)v'_i = 0$
for all $i$, then $m(T)$ is the minimal polynomial of $M$ over $\mathbb{Q}$.
This follows because if $m(M)v'_i = 0$ then $m(M)$ annihilates all of $W'_i$ for each $i$.
Thus it annihilates all of $V'$.

The cost of computing the $m(M)v'_i$ is usually low
compared to computing $m(M)$ directly, since it consists of matrix-vector
rather than matrix-matrix multiplications.

In the worst case this algorithm requires $O(n^4)$ operations over $\mathbb{Z}$
(once we encounter a good prime), but this is far from the generic case, even when
the minimal polynomial is not the characteristic polynomial.

In practice our multimodular algorithm seems to slightly outperform Magma on the
examples we have tried, including matrices with minimal polynomial smaller 
than the characteristic polynomial.
A generic version of the algorithm over fields is implemented in Julia code in
Nemo, and an efficient version over $\mathbb{Z}$ using the Chinese remaindering
trick is implemented in Flint and made available in Nemo.

\section{C/C++ wrappers in Nemo}
\label{sect:wrappers}

\subsection{Flint: arithmetic and number theory}

Flint provides arithmetic over
$\mathbb{Z}$, $\mathbb{Q}$, $\mathbb{Z}/n\mathbb{Z}$,
$\operatorname{GF}(q), \mathbb{Q}_p$
as well as matrices, polynomials and power series over most
of these ground rings.
Nemo implements elements of these rings as thin wrappers of the
Flint C types.

Flint uses a number of specialised techniques for each domain.
For example, Flint's multiplication in $\mathbb{Z}[x]$
uses a best-of-breed algorithm which selects classical multiplication,
Karatsuba multiplication, Kronecker segmentation or a 
Sch\"{o}nhage-Strassen FFT, with accurate cutoffs
between algorithms, depending on the degrees and coefficient sizes.

In some cases, Flint provides separate implementations
for word-size and arbitrary-size coefficients.
Nemo wraps both versions and transparently selects the
optimised word-size version when possible.

Additional Flint algorithms wrapped by Nemo include
primality testing, polynomial factorisation, LLL, and Smith and Hermite
normal forms of integer matrices.

\subsection{Arb: arbitrary precision ball arithmetic}

Computing over $\mathbb{R}$ and $\mathbb{C}$ requires
using some approximate model for these fields.
The most common model is floating-point arithmetic.
However, for many computer algebra algorithms,
the error analysis necessary to guarantee correct
results with floating-point arithmetic
becomes impractical.
Interval arithmetic solves this problem by effectively making
error analysis automatic.

Nemo includes wrapper code for the C library Arb, which implements real numbers as
arbitrary-precision midpoint-radius intervals (balls) $[m \pm r]$
and complex numbers as rectangular boxes $[a \pm r_1]$ + $[b \pm r_2] i$.
Nemo stores the precision in the parent object.
For example, \texttt{R = ArbField(53)}
creates a field of Arb real numbers with 53-bit precision.
Arb also supplies types for polynomials, power series and matrices
over $\mathbb{R}$ and $\mathbb{C}$, as well as transcendental functions.
Like Flint, Arb systematically uses asymptotically fast algorithms
for operations such as polynomial multiplication, with tuning
for different problem sizes.

Many problems can be solved using lazy evaluation: the user can
try a computation with some tentative precision $p$ and restart
with precision $2p$ if that fails. The precision can be set
optimally when a good estimate for the minimal
required~$p$ is available; that is, the intervals
can be used as if they were plain floating-point numbers, and the automatic
error bounds simply provide a certificate.

Alternative implementations of $\mathbb{R}$ and $\mathbb{C}$
may be added to Nemo in the future.
For example, it would sometimes be more convenient to use a lazy
bit stream abstraction in which individual numbers
store a symbolic DAG representation to allow automatically
increasing their own precision.


\subsection{Antic: algebraic number fields}

Antic is a C library, depending on Flint, for
doing efficient computations in algebraic number fields. 
Nemo uses Antic for number field element arithmetic. We briefly describe some
of the techniques Antic uses for fast arithmetic, but
refer the reader to the article \cite{antic} for full details.

Antic represents number field elements as Flint polynomials thereby benefiting
from the highly optimised polynomial arithmetic in Flint. However, a few
more tricks are used.

Firstly, Antic makes a number of precomputations which are stored in a context
object to speed up subsequent computations with number field elements. Number field parent
objects in Nemo consist precisely of these context objects.

The first is a precomputed inverse of the leading coefficient of the defining
polynomial of the number field. This helps speed up reduction modulo the defining polynomial.

The second is an array of precomputed powers $x^i$ modulo the
defining polynomial $f(x)$. This allows fast reduction of polynomials whose degree exceeds
that of the defining polynomial, e.g.\ in multiplication of number field elements $g$ and $h$
where one wants wants to compute $g(x)h(x) \pmod{f(x)}$.

The third precomputation is of Newton sums $S_k = \sum_{i=1}^n \theta_i^k$ where the $\theta_i$
are the roots of the defining polynomial $f(x)$. These Newton sums are precomputed using
recursive formulae as described in \cite{Cohen1993}. They are
used to speed up the computation of traces of elements of the number field.

Norms of elements of number fields are computed using resultants, for which we use
the fast polynomial resultant code in Flint. Inverses of elements are computed using the fast polynomial
extended GCD implementation in Flint.

Antic also offers the ability to do multiplication of number field elements without
reduction. This is useful for speeding up the dot products that occur in matrix
multiplication, for example. Instead of reducing after every multiplication, the unreduced
products are first accumulated and their sum can be reduced at the end of the dot product.

To facilitate delayed reduction, all Antic number field elements are allocated with
sufficient space to store a full polynomial product, rather than the reduction of such a product.

\subsection{Singular: commutative algebra}

Singular \cite{singular} is a C/C++ package for polynomial systems and algebraic
geometry. Recently, we helped prepare a Julia package called Singular.jl, which is compatible with
Nemo. It will be described in a future article.

\section{Hecke: algebraic number theory in Julia}
\label{sect:hecke}

Hecke is a tool for algebraic number theory, written in Julia.
Hecke includes the following functionality:
\begin{itemize}
\setlength{\itemsep}{1pt}
\setlength{\parskip}{0pt}
\setlength{\parsep}{0pt}
\item
  element and ideal arithmetic in orders of number fields,
\item
  class group and unit group computation,
\item
  verified computations with embeddings and
\item
  verified residue computation of Dedekind zeta functions.
\end{itemize}

Hecke is written purely in Julia, though it makes use of Flint, Arb and Antic through the interface provided by Nemo.
Hecke also applies the generic constructions and algorithms provided by Nemo to its own types. This allows for example to
work efficiently with polynomials or matrices over rings of integers without the necessity to define special types for these objects.

For most computational challenges we rely on well known techniques as described in \cite{Cohen1993, Pohst1997, Belabas2004}
(with a few modifications). However, we use new strategies for ideal arithmetic and computations with approximations.

\subsection{Fast ideal arithmetic in rings of integers}

A classical result is that in Dedekind domains, as for the ring of integers of a number field, any
ideal can be generated using only two elements, one of which can be chosen
at random in the ideal. This representation is efficient in terms of space compared to storing a
$\ZZ$-basis. However, the key problem is the lack of
efficient algorithms for working with two generators. We remark that one
operation is always efficient using two generators: powering of ideals.

A refinement of this idea, due to Pohst and Zassenhaus \cite[p. 400]{Pohst1997}, is a 
\textit{normal} presentation. Here a finite set $S$ of prime numbers
is fixed, typically containing at least all prime divisors of the norm of the ideal. Then,
based on this set, the two generators are chosen: The first as an integer
having only divisors in $S$---but typically with too high multiplicities.
The second generator then has the correct multiplicity at all prime ideals
over primes in $S$---but is allowed to have other divisors as well.

For the remainder of this section, let $K/\QQ$ be a number field of
degree $n$ and $\mathcal O_K$ be its ring of integers.

\begin{definition}
Let $S$ be a finite set of prime numbers and $A$ a nonzero ideal such that
all primes $p$ dividing the norm $N(A) = |\mathcal O_K/A|$ are in $S$.
A tuple $(a, \alpha)$ is an $S$-normal presentation for $A$ if and only if
\begin{itemize}
\setlength{\itemsep}{1pt}
\setlength{\parskip}{0pt}
\setlength{\parsep}{0pt}
\item we have $a\in A\cap \mathbb Z$, $\alpha\in A$,
\item for all prime ideals $Q$ over $p\not\in S$ we have $v_Q(a)= v_Q(A) = 0$
  for the exponential valuation $v_Q$ associated to $Q$,
\item for all prime ideals $Q$ over $p\in S$ we have $v_Q(\alpha) = v_Q(A)$.
\end{itemize}
\end{definition}

A direct application of the Chinese remainder theorem shows the existence of such normal presentations.
The algorithmic importance comes from the following result.

\begin{theorem}
Let $A = \langle a, \alpha \rangle$ be an ideal in $S$-normal presentation. Then the following hold:
\begin{enumerate}
\item We have $N(A) = \gcd(a^n, N(\alpha))$.
\item If $B = \langle b, \beta\rangle$ is a second ideal in $S$-normal
presentation for the same set $S$, then
$AB = \langle ab, \alpha\beta\rangle$ is also a $S$-normal presentation.
\item Let $\gamma = 1/\alpha$, $\langle g\rangle = \alpha \mathcal O_K \cap \ZZ$,
 $g = g_1g_2$ with coprime $g_1$, $g_2$ and $g_2$ only divisible by primes not in $S$. Then $A^{-1} = \langle 1, g_2\gamma\rangle$ is an $S$-normal presentation.
\item If $p$ is a prime number and $P= \langle p, \theta\rangle$ a prime ideal over $p$
in $\{p\}$-normal presentation, then the $\{p\}$-normal presentation
of $P^{-1} = \langle 1, \gamma\rangle$ yields a valuation element, that is,
$v_P(\beta) = \max\{ k \mid \gamma^k\beta\in \mathcal O_K\}$ for all $\beta \in \mathcal O_K$.
\end{enumerate}
\end{theorem}

It should be noted that (almost all) prime ideal are naturally computed
in $\{p\}$-normal presentation. The key problem in using the theorem for
multiplication is that both ideals need to be given by an $S$-normal presentation for
the same set $S$, which traditionally is achieved by recomputing
a suitable presentation from scratch, taking random candidates for the
second generator until it works. Based on the following lemma, we have developed a new algorithm that manages it at the cost 
of a few integer operations only.

\begin{lemma}
Let $S$ be a finite set of prime numbers and $A=\langle a, \alpha\rangle$ an ideal in
$S$-normal presentation. Let $T$ be a second set of primes and set
$s = \prod_{x\in S} x$ as well as $t = \prod_{x\in T\setminus S} x$. If $1 = \gcd(as, t) = uas + vt$,
then $\langle a, \beta \rangle$ is an $S\cup T$-normal presentation, where $\beta = vt\alpha + uas$.
\end{lemma}

\begin{proof}
By definition, $a$ has only prime divisors in $S$, so $as$ and $t$ are coprime.
Also, $a$ is a possible first generator for the $S\cup T$-normal presentation.
Let $P$ be a prime ideal over a prime $p\in S$, then
$v_P(\beta) = v_P(\alpha)= v_P(A)$ since $v_P(A) <v_P(s)+v_P(a)$ due to $v_P(s)\ge 1$ and $v_P(a) \ge v_P(A)$ by definition.
For $P$ coming from $T\setminus S$, we have $v_P(\beta) = 0$ since $v_P(t)\ge 1$ and $v_P(uas) = 0$ as well.
\end{proof}

Using this lemma on both input ideals, we can obtain compatible $S$-normal 
presentations at the cost of two $\gcd$ computations, a few integer multiplications and two products of integers with algebraic numbers. The final ideal
multiplication is then a single integer product and a product of
algebraic numbers. Thus the asymptotic cost is that of a single multiplication of two algebraic numbers.
In contrast, all previous algorithms require a linear algebra step.

Finally, we improve on the computation of an $S$-normal presentation.

\begin{lemma}
If $0\ne\alpha\in \mathcal O_K$, then $\langle\alpha\rangle\cap \mathbb Z = \langle d\rangle$ where $d>0$ is minimal such that $d/\alpha\in \mathcal O_K$.
\end{lemma}

\begin{theorem}[\cite{Pohst1997}]
Let $A$ be a nonzero ideal. Then the tuple
$(a, \alpha)$ is an $\{p \, | \, p \text{ divides }a\}$-normal presentation of $A$ if and only if
$\gcd(a, {d}/{\gcd(a, d)}) = 1$, where $d = \min(\alpha\mathcal O_K \cap \NN)$.
\end{theorem}

Together with the above lemma, this allows for the rapid computation of a
normal presentation: Choose $\alpha\in A$ at random until a normal presentation
is obtained. It can be shown that, unless $a$ involves many ideals of small
norm, this is very efficient.

To illustrate the speed of our algorithm, we created a list of ideals of bounded norm (here: 400) and
took random products of 100 ideals in the field
defined by the polynomial $X^n + 2$ for $n=16, 32, 64, 128$.
The results are presented in Table~\ref{tab:ideals}.
Times are given in seconds.

\begin{table}
\center
\caption{Time (s) on ideal multiplication.}
\begin{small}
\setlength{\tabcolsep}{2.0pt}
\renewcommand{\arraystretch}{0.95}
\begin{tabular}{c c c c} \hline
$n$ & Magma & Pari & Hecke  \\ \hline
     16  &  8.44 & 0.05 & 0.02   \\
     32  &   235.82  &   0.18   &     0.04 \\
     64  &   905.44  &     0.96   &     0.06    \\
    128  &   7572.19   &     5.40    &       0.08 
\end{tabular}
\label{tab:ideals}
\end{small}
\end{table}


\subsection{The use of interval arithmetic}

One is often forced to study algebraic number fields in a real or complex setting, e.g.\ embeddings into an algebraically
closed field, or computing Dedekind zeta functions.
This can be due to an intrinsic property of the problem, or it may be the fastest (known)
 way to solve the problem. Either way, the price is working with approximations.
We give a few examples of this and show how the ball arithmetic provided by Arb is employed in Hecke for this purpose.

\subsubsection{Computing conjugates}\label{subsub:conj}

Let $K = \QQ[X]/(f)$ be an algebraic number field, where $f \in \QQ[X]$ is an irreducible polynomial of degree $d$.
We represent the elements of $K$ as polynomials of degree less than $d$.
Denoting by $\alpha_1,\dotsc,\alpha_d \in \CC$ the roots of $f$ in $\CC$, the distinct embeddings $K \to \CC$ are given by
$\sigma_i \colon K \to \CC, \bar X \mapsto \alpha_i$. For an element $\alpha$ of $K$ the complex numbers $\sigma_i(\alpha)$, $1 \leq i \leq d$
are called the conjugates of $\alpha$. Since for $\alpha \in K \setminus \QQ$ the conjugates are irrational, it is clear that we must rely on
approximations.

In Hecke this is done using ball arithmetic. Let $\alpha = \sum a_j X^j$ be an element of $K$. 
Assume that we want to find $\hat \sigma_i(\alpha) \in \RR$ with $\lvert \sigma_i(\alpha) - \hat \sigma_i(\alpha) \rvert \leq 2^{-p}$
for some precision $p \in \ZZ_{\geq 1}$.
Using Arb's polynomial root finding functionality and some initial precision $p' \geq p$ we find balls $\smash{B_{p'}^{(i)} \subseteq \RR}$ such
that $\smash{\alpha_i \in B_{p'}^{(i)}}$ and $\smash{\operatorname{diam}(B_{p'}^{(i)}) \leq 2^{-p'}}$.
Ball arithmetic then yields balls $B_\alpha^{(i)} \subseteq \RR$ with
\[ \sigma_i(\alpha) = \sum_{1 \leq j \leq d} a_j \alpha_i^j \in \sum_{1 \leq j \leq d} a_j (B_{p'}^{(i)})^j\subseteq B_\alpha^{(i)}. \]

Finally we check whether $\smash{\operatorname{diam}(B_\alpha^{(i)}) \leq 2^{-p}}$. If not, we increase the working precision $p'$ and repeat.
When finished, the midpoints of the balls $\smash{B_\alpha^{(i)}}$ are approximations $\hat \sigma_i(\alpha)$ with the desired property.
The following Hecke code illustrates this.

\begin{small}
\begin{verbatim}
QQx, x = PolynomialRing(QQ, "x")
K, a = NumberField(x^3 + 3*x + 1, "a")
alpha = a^2 + a + 1
p = 128; p' = p

while true
  CCy, y = PolynomialRing(AcbField(p'), "y")
  g = CCy(x^3 + 3*x + 1)
  rts = roots(g)
  z = map(x^2 + x + 1, rts)
  if all([ radiuslttwopower(u, -p) for u in z])
    break
  else
    p' = 2*p'
  end
end
\end{verbatim}
\end{small}

To perform the same task using floating point arithmetic would require a priori error analysis.
Arb's ball arithmetic allows the error analysis to be carried along with the computation. This allows for fast development,
while at the same time maintaining guaranteed results.

\subsubsection{Torsion units}

For an algebraic number field $K$ of degree $d$, the set $\{ \alpha \in K^\times \mid \exists n \in \ZZ_{\geq 1} \colon \alpha^n = 1 \}$ of
torsion units is a finite cyclic subgroup of $K^\times$, which plays an important role
when investigating the unit group of the ring of integers of $K$.

Given $\alpha \in K^\times$ it is important to quickly check if it is a torsion unit.
A list of possible integers $n$ with $\alpha^n = 1$ can be obtained as follows:
If $\alpha$ is a torsion unit, then $\alpha$ is a primitive $k$-th root of unity for some $k$.
In particular $K$ contains the $k$-th cyclotomic field and $\varphi(k)$ divides $n$.
Since there are only finitely many $k$ with the property that $\varphi(k)$ divides $n$,
for every such $k$ we can test $\alpha^k = 1$ and in this way decide whether $\alpha$ is a torsion unit.
While this works well for small $n$, for large $n$ this quickly becomes inefficient.

A second approach rests on the fact that torsion units are characterized by their conjugates: An element
$\alpha$ is a torsion unit if and only if $\lvert \sigma_i(\alpha) \rvert = 1$ for $1 \leq i \leq d$.
Although we can compute approximations of conjugates with arbitrary precision, since conjugates are in general
irrational, it is not possible test whether they (or their absolute value) are exactly equal to $1$.

We make use of the following result of Dobrowolski~\cite{Dobrowolski1978}, which bounds the modulus of conjugates of non-torsion units.

\begin{lemma}
If $\alpha$ is not a torsion unit, then there exists some $1 \leq i \leq k$ such that
\[ \left(1 + \frac 1 6 \frac{\log(d)}{d^2} \right) < \lvert \sigma_k(\alpha) \rvert. \]
\end{lemma}

We can rephrase this using approximation as follows.

\begin{lemma}\label{lem:tor}
  Let $(\!B^{(i)}_k)_{k \geq 1}$, $\!1\!\!\leq\!i\!\!\leq\!\!d$, be sequences of real balls with $\lvert \sigma_i(\alpha) \rvert \in B^{(i)}_k$ and $\max_i\operatorname{diam}(B_k^{(i)}) \to 0$ for $k \to \infty$.
  \begin{enumerate}
  \item
    If the element $\alpha$ is torsion, there exists $k \geq 1$ such that $\smash{1 < B_k^{(i)}}$ for some $1 \leq i \leq d$.
  \item
    If the element $\alpha$ is non-torsion, there exists $k \geq 1$ such that $\smash{B_k^{(i)} < 1 + \log(d)/(6d^2)}$ for all $1 \leq i \leq d$.
  \end{enumerate}
\end{lemma}

It is straightforward to turn this into an algorithm to test whether $\alpha$ is torsion: Approximate the conjugates of $\alpha$ using balls
as in~\ref{subsub:conj} for some starting precision $p \geq 1$ and check if one of the statements of Lemma~\ref{lem:tor} hold.
If not, increase the precision and start over.
Since the radii of the balls tend to $0$, by Lemma~\ref{lem:tor} the algorithm will eventually terminate.

\subsubsection{Residues of Dedekind zeta functions}

For a number field $K$ and $s \in \CC$, $\operatorname{Re}(s) > 1$, the Dedekind zeta function of $K$ is defined as
\[ \zeta_K(s) = \sum_{\{0\} \neq I \subseteq \mathcal O_K} \frac 1 {N(I)^s}, \]
where the sum is over all nonzero ideals of the ring of integers $\mathcal O_K$ of $K$.
This function has an analytic continuation to the whole complex plane with a simple pole at $s = 1$. The analytic class number formula
\cite{Cohen1993} states that the residue at that pole encodes important arithmetic data of $K$:
\[ \operatorname{Res}(\zeta_K, 1) =\lim_{s \to 1}(s - 1)\zeta_K(s) = h_K \cdot \mathrm{reg}_K \cdot c_K, \]
where $h_K$ is the class number, $\mathrm{reg}_K$ is the regulator and $c_K$ is another (easy to determine) constant depending on $K$.

The analytic class number formula is an important tool during class group
computations. By approximating the residue, it allows one to certify that
tentative class and unit groups are in fact correct (see \cite{Biasse2014}).

We describe the approximation under GRH using an algorithm of Belabas and Friedmann~\cite{Belabas2015}, but it is also possible with
results from Schoof~\cite{Schoof1982} or Bach~\cite{Bach1995}.
Since the aim is to illustrate the use of ball arithmetic, we will not give detailed formulas for the approximation or the error.

\begin{theorem}[Belabas-Friedmann]
There exist functions $g_K$, $\varepsilon_K \colon \RR \to \RR$, with
\[ \lvert \log(\operatorname{Res}(\zeta_K, 1)) - g_K(x) \rvert \leq \varepsilon_K(x) \]
for all $x \geq 69$. The evaluation of $g_K(x)$ involves only prime powers $p^m \leq x$ and prime ideal powers $\mathfrak p^m$ with
$N(\mathfrak p^m) \leq x$. Moreover the function $\varepsilon_K$ is strictly decreasing and depends only on the degree and the
discriminant of $K$.
\end{theorem}

Assume that $\varepsilon > 0$ and we want to find $z \in \RR$ such that
\[ \lvert \log(\operatorname{Res}(\zeta_K, 1)) - z \rvert \leq \varepsilon.\]
We first compute an $x_0$ such that $\varepsilon_K(x_0) < \varepsilon/2$. Note that since $\varepsilon_K$ is strictly decreasing,
this can be done using ordinary floating point arithmetic with appropriate rounding modes.
Once we have this value we compute a real ball $B$ with $g_K(x_0) \in B$ and $\operatorname{diam}(B) \leq \varepsilon/2$ (as usual
we progressively increase precision until the radius of the ball is small enough).
By the choice of $x_0$, the midpoint of $B$ is an approximation to the logarithm of the residue with error at most $\varepsilon$.
Again this yields a guaranteed result, but avoids the tedious error analysis due to the form of $g_K$.

\section{Acknowledgments}

This work was supported by DFG priority project SPP 1489, DFG SFB-TRR 195 and ERC Starting Grant ANTICS 278537.

\end{document}